\documentclass[letterpaper, 10 pt, conference]{ieeeconf}
\IEEEoverridecommandlockouts                              % This command is only needed if 
\usepackage{amsmath,amssymb,amsfonts}
\usepackage{algorithm}
\usepackage{algorithmic}
\usepackage{cite}
\usepackage{algorithmic}
\usepackage{graphicx}
\usepackage{textcomp}
\usepackage{xcolor}
\usepackage{caption}
\usepackage{verbatim}
\usepackage{mathtools}
\usepackage{derivative}

\newtheorem{theorem}{Theorem}
\newtheorem{lemma}{Lemma}
\newtheorem{proposition}{Proposition}

\newtheorem{definition}{Definition}
\newtheorem{remark}{Remark}

\newcommand{\bmtx}{\begin{bmatrix}}
\newcommand{\emtx}{\end{bmatrix}}
\newcommand{\bsmtx}{\left[ \begin{smallmatrix}} 
\newcommand{\esmtx}{\end{smallmatrix} \right]}
\newcommand{\field}[1]{\mathbb{#1}}
\newcommand{\R}{\field{R}}

\newcommand{\mcl}[1]{\mathcal{#1}}
\newcommand{\Cset}{\mathcal{C}}
\newcommand{\mbf}[1]{\mathbf{#1}}
\newcommand{\argmin}{\operatornamewithlimits{arg\,min}}

\DeclareMathOperator{\vol}{vol}

\def\BibTeX{{\rm B\kern-.05em{\sc i\kern-.025em b}\kern-.08em
    T\kern-.1667em\lower.7ex\hbox{E}\kern-.125emX}}

\title{\textbf{Backup Control Barrier Function Synthesis using \\ Sum-of-Squares Reachability}}

\author{Jungbae Chun$^{1}$, Shima Sadat Mousavi$^{2}$, David E. J. van Wijk$^{2}$, \\ Ersin Da\c{s}$^{3}$, Aaron D. Ames$^{2}$, and Felix Biert\"umpfel$^{4}$  %
\thanks{This research was supported by the NSF under grants No. 2337751 and 2337752, and by the EU under Grant No. 101153910. 
}
\thanks{$^{1}$Department of Electrical and Computer Engineering, University of Michigan, Ann Arbor, MI 48109, USA, 
\texttt{jungbaec@umich.edu}.}
\thanks{$^{2}$Mechanical and Civil Engineering, California Institute of Technology, Pasadena, CA 91125, USA, \texttt{\{smousavi, vanwijk, ames\}@caltech.edu}.}
\thanks{$^{3}$Department of Mechanical, Materials, and Aerospace Engineering, Illinois Institute of Technology, Chicago, IL 60616, USA, \texttt{ edas2@illinois.edu}.}
\thanks{$^{4}$Department of Aerospace Engineering, Auburn Univeristy,  Auburn, AL 36849, USA, \texttt{feb0013@auburn.edu}.}
}

\begin{document}
\captionsetup{font=footnotesize}

\maketitle
\begin{abstract}
Backup control barrier functions (bCBFs) enforce safety for input-constrained nonlinear systems using a pre-certified backup set and controller, but their performance depends strongly on this prescribed pair. This letter develops a constructive method for synthesizing a less conservative backup pair via finite horizon sum-of-squares (SOS) backward reachability. Starting from an initial backup set, we compute an SOS-certified finite horizon backward reachable set and controller that satisfy safety and input constraints while steering trajectories to the original backup set. We then provide conditions under which this certified set becomes a valid backup set for a piecewise backup controller. The resulting backup pair is integrated into the bCBF framework to certify larger safe sets.
\end{abstract}

\section{Introduction}

Safety filters based on control barrier functions (CBFs) provide an effective way to enforce state constraints online while minimally modifying a nominal controller \cite{Ames2017}. For  systems with bounded inputs, however, the main difficulty is often the construction of a sufficiently large controlled invariant set on which safety can be certified, and feasibility can be maintained \cite{mousavi2025vertices,Mousavi2026Structure}. Backup control barrier functions (bCBFs) address this issue by propagating a backup controller to a forward invariant backup set and using the resulting predicted flow to define an implicit safe set \cite{Gurriet2018,Gurriet2020,Chen2021bCBF}. Yet the certified region obtained in this way is fundamentally limited by the chosen backup set--backup controller pair. If the backup set is very small or the backup controller is conservative, then the resulting safety filter is conservative. This raises a central question for input-constrained safety-critical control: \textit{Can we synthesize a larger backup set, together with a backup controller that certifies it, instead of relying only on potentially conservative backup pairs?}

A substantial body of work has developed CBF-based safety filtering and its extensions to systems with input bounds. Classical CBF-QP methods provide forward invariance guarantees for explicitly known safe sets \cite{Ames2017}, while input-constrained variants and related formulations seek to account for actuator limitations directly in the CBF construction \cite{Agrawal2021ICCBF}. Backup CBFs offer an attractive alternative because they construct controlled invariant sets implicitly through forward prediction of the flow of the system under a backup controller, thereby retaining tractability for nonlinear systems with bounded inputs \cite{Gurriet2018,Gurriet2020,Chen2021bCBF,Molnar2023ROM}. Recent works have further expanded this line of research to disturbance-robust formulations, uncertainty-estimator-based robustification, and generalized set-expanding controllers \cite{vanWijk2024DR,vanWijk2025UE,vanWijk2026Gen}. However, the existing methods still presuppose the availability of a valid backup pair. Even recent constructive approaches remain tied to particular structural assumptions, such as feedback linearization or prescribed local backup designs \cite{Gacsi2025}.

In parallel, reachability- and sum-of-squares (SOS)-based methods provide a complementary route to reducing conservatism. Hamilton--Jacobi reachability and viability theory can characterize safe operating regions with strong guarantees, but their computational cost limits scalability for many nonlinear systems \cite{Mitchell2005,Aubin1991,SaintPierre1994}. On the other hand, polynomial funnel and SOS methods enable tractable certificate construction without gridding \cite{Majumdar2013,Majumdar2017}. Particularly relevant here is the finite horizon polynomial backward-reachability framework of \cite{yin2021backward}, which synthesizes both an admissible controller and a certified inner approximation of the backward reachable set under state and input constraints. Closely related works have also sought to bridge CBFs and reachability, for example, through control barrier-value functions, Hamilton--Jacobi refinement of candidate barrier functions, and finite horizon prediction-based CBF construction \cite{Choi2021CBVF,Tonkens2022,Wiltz2023}. However, these approaches stop short of showing that a backward reachable set can itself serve as a new backup set in the bCBF sense, namely, as a forward invariant set under an associated backup controller.

This letter fills that gap. Starting from an initial backup set and backup controller, we use finite horizon SOS backward reachability to construct a polynomial backward reachable set and an associated controller subject to safety and input constraints. Under additional invariance and containment conditions, this SOS-certified set becomes a valid backup set for a constructed piecewise backup controller. In this way, the proposed method facilitates the synthesis of an enlarged backup set--backup controller pair, without assuming the final pair is given a priori. The paper makes three key contributions: 1) an SOS formulation for constructing a backward reachable set and associated controller under safety and input constraints; 2) sufficient conditions under which the SOS-certified set becomes a valid backup set, and 3) an integration of the synthesized backup pair into the backup-CBF framework, yielding less conservative safety filtering than approaches that assume an a priori specified initial backup set--backup controller pair.

\paragraph*{Notation}
For ${\xi\in\R^n}$, let ${\R[\xi]}$, ${\R^m[\xi]}$, and ${\R^{m\times n}[\xi]}$ denote scalar-, vector-, and matrix-valued polynomials in $\xi$, respectively. The set of sum-of-squares (SOS) polynomials is denoted by
${\Sigma[\xi] := \{\pi=\sum_{i=1}^{M}\pi_i^2:
M\geq 1,\ \pi_i\in\R[\xi]  \}}$,
which is a subset of $\R[\xi]$. For a continuous function ${r:\R^n \to \R}$ and scalar ${\eta \in \R}$, define the sublevel set ${\Omega_\eta^r:=\{x \in \R^n : r(x) \le \eta\}}$. Similarly, for a time-dependent continuous function ${r:\R\times\R^n\to\R}$, define ${\Omega_{t,\eta}^r:= \{x \in \R^n : r(t,x) \le \eta\}}$.

\section{Preliminaries}
\label{sec:setup}

We consider nonlinear control-affine systems of the form
\begin{equation}
\dot{x} = f(x) + g(x)u,
\label{eq:sys_dynamics}
\end{equation}
where $x \in \mathcal{X} \subseteq \R^n$ is the state, $u \in \mathcal{U} \subseteq \R^m$ is the control input, and $\mathcal{U}$ is a convex polytope. The functions $f : \mathcal{X} \to \R^n$ and $g : \mathcal{X} \to \R^{n \times m}$ are assumed to be $C^1$ and either polynomial or replaced by polynomial approximations on $\mathcal{X}$, i.e., $f \in \R^n[x]$ and $g \in \R^{n \times m}[x]$. 
In the remainder of the paper, $\mathcal{X}$ is assumed compact and semialgebraic.

We aim to synthesize a feedback controller ${k:\mcl{X}\rightarrow \mcl{U}}$, with ${k \in \R^m[x]}$, such that the closed-loop system from the feedback interconnection of $k$ and~\eqref{eq:sys_dynamics} is guaranteed safe.

\subsection{Control Barrier Functions}
Safety is formulated with respect to a set $\Cset_{\rm S}\subseteq \mcl{X}$
and a $\Cset^1$ safety function $h: \mcl{X}\rightarrow \R$ as $\Cset_{\rm S}:=\left\{x \in \mcl{X}: h(x)\ge 0\right\}$. For a system to be considered safe, $\Cset_{\rm S}$ must be forward invariant for the closed-loop system built from~\eqref{eq:sys_dynamics}. In other words, if~$x_0 := x(0) \in \Cset_{\rm S}\Rightarrow x(t)\in \Cset_{\rm S}$ for all $t \ge 0$.

Conditions for forward invariance are provided by control barrier functions (CBFs) that enable the synthesis of safe controllers. A function $h$ is a control barrier function if there exists a class-$\mcl{K}_\infty$ function $\alpha$ such that
\begin{equation}
    \label{eq:cbf_condition}
    \sup_{u\in\mcl{U}} \dot{h}(x, u) \ge -\alpha(h(x)),\,\forall x\in \Cset_{\rm S}.
\end{equation}
The following theorem formalizes this concept.
\begin{theorem}[Control Barrier Functions~\cite{Ames2017}]\label{thm:cbf}
If~$h$ is a control barrier function for the system~\eqref{eq:sys_dynamics}, then any locally Lipschitz controller~$k:\mcl{X}\rightarrow\mcl{U}$ with $u=k(x)$ satisfying
\begin{equation}\label{eq:cbfcon}
    \dot{h}(x,u)\ge -\alpha(h(x)),\,\forall x\in\Cset_{\rm S}
\end{equation}
renders the set~$\Cset_{\rm S}$ forward invariant.
\end{theorem}

Safety can be encoded into an existing primary controller~$k_{\rm p}:\mcl{X}\rightarrow \mcl{U}$ using the solution of a quadratic program (QP) that is constrained by~\eqref{eq:cbfcon}. 
The QP solution provides a minimally altered yet safe control signal.
However, finding a safety function $h$ that is a CBF under input constraints, i.e., $h$ satisfies the CBF condition \eqref{eq:cbf_condition}, is challenging. The main reason for this is that the set $\Cset_{\rm S}$ must be controlled invariant, i.e., there exists a controller ${k:\mcl{X}\rightarrow \mcl{U}}$ with $u = k(x)$ rendering $\mathcal{C}_{\rm S}$ forward invariant for \eqref{eq:sys_dynamics}.

\subsection{Backup Control Barrier Functions}
\label{sec:backup}
Constructive methods for finding valid CBFs exist but are often conservative, i.e., the corresponding safe sets are small, thereby limiting the performance of the safe controller \cite{Chen2021bCBF}. It is possible to expand such a conservative safe set by introducing backup controllers associated with a backup control barrier function (bCBF) \cite{Gurriet2018, Gurriet2020}.
The goal of the paper is to advance the synthesis of an enlarged backup set--backup controller pair.
The next subsections revisit the necessary background on bCBFs and SOS programs. 

We first define the backup set and backup controller.
\begin{definition}[Backup Set and Controller]\label{def:bCBF}
~A continuously differentiable controller~$k_{\rm b}: \mcl{X}\rightarrow \mcl{U}$ is called a \textit{backup controller} if it renders
the \textit{backup set}
\begin{equation}\label{eq:Cb}
\Cset_{\rm B}:=\{x\in \mcl{X}: h_{\rm b}(x) \ge 0\}\subseteq \Cset_{\rm S},
\end{equation}
%$\Cset_{\rm B}:=\{x\in \mcl{X}: h_{\rm b}(x) \ge 0\}\subseteq \Cset_{\rm S}$ 
forward invariant for \eqref{eq:sys_dynamics}. The function~$h_{\rm b}:\mcl{X}\rightarrow\R$ is $C^1$ with zero as regular value, i.e., $\nabla h_{\rm b}\neq 0$ for all $x$ on the boundary of the backup set.
\end{definition}

The bCBF method expands \(\Cset_{\rm B}\) by evolving system~\eqref{eq} under the backup controller \(k_{\rm b}\). The resulting flow \(\phi_{\rm b}(\tau,x)\), \(\tau\in[0,T]\), with fixed \(T>0\), is defined by
\begin{equation}\label{eq:flow_de}
\tfrac{\partial }{\partial \tau}\phi_{\rm b}(\tau, x) = f_{\rm b}(\phi_{\rm b}(\tau, x)),\quad \phi_{\rm b}(0, x) = x,
\end{equation}
where $f_{\rm b}(x) := f(x) + g(x)k_{\rm b}(x)$. We can use this flow to define the expanded set $\Cset_{\rm I}$ as
\begin{equation}\label{eq:Ci}
\Cset_{\rm I}:= \left\{ x\in \mcl{X} :  \begin{aligned}
&h(\phi_{\rm b}(\tau,x)) \ge 0,\ \forall \tau \in [0,T],\\
&h_{\rm b}(\phi_{\rm b}(T,x)) \ge 0
\end{aligned}\right\} \subseteq \Cset_{\rm S},
\end{equation}
which contains the states from which the system~\eqref{eq:sys_dynamics} using the backup controller~$k_{\rm b}$ safely reaches the backup set $\Cset_{\rm B}$ in the finite time $T$. Note that~$\Cset_{\rm I}$ is defined such that the flow of the system remains within~$\Cset_{\rm S}$ for all times~$\tau \in [0, T]$.  The set~$\Cset_{\rm I}$ is the expansion of $\Cset_{\rm B}$ as $\Cset_{\rm I} \equiv \Cset_{\rm B}$ if~$T = 0$. Due to this expansion, the set~$\Cset_{\rm I}$ is controlled invariant by construction as formalized in Lemma~\ref{lem:CI}.
\begin{lemma}[\!\!\cite{Gurriet2020}]\label{lem:CI}
The set~$\Cset_{\rm I}$ defined in~\eqref{eq:Ci} is controlled invariant, and the backup controller~${k_{\rm b}:\mcl{X}\rightarrow \mcl{U}}$ renders~$\Cset_{\rm I}$ forward invariant for~\eqref{eq:sys_dynamics}.
\end{lemma}
\begin{proposition}
[\!\!\cite{Gurriet2020}]\label{thm:hb}
There exist class-$\mathcal{K}_\infty$ functions $\alpha, \alpha_{\rm b}$, and a locally Lipschitz controller~$k:\mcl{X}\rightarrow \mcl{U}$ with~$u = k(x)$ satisfying
\begin{subequations}
\begin{align}
\begin{split}\label{eq:hb1}
\nabla h(\phi_{\rm b}(\tau, x)) \frac{\partial \phi_{\rm b}(\tau, x)}{\partial x} \dot{x} &\ge -\alpha(h(\phi_{\rm b}(\tau,x))),\\
&\quad \forall\, \tau \in [0,T],
\end{split}\\
\nabla h_{\rm b}(\phi_{\rm b}(T, x)) \frac{\partial \phi_{\rm b}(T, x)}{\partial x} \dot{x} &\ge -\alpha_{\rm b}(h_{\rm b}(\phi_{\rm b}(T,x))), 
\label{eq:hb2}
\end{align}
\end{subequations}
for all~$x\in \Cset_{\rm I}$. The same controller renders $\Cset_{\rm I} \subseteq\Cset_{\rm S}$ forward invariant, and satisfies the input constraints.
\end{proposition}

Proposition~\ref{thm:hb} yields a QP-based safety filter that enforces safety of~\eqref{eq:sys_dynamics} on the expanded set~$\Cset_{\rm I}$:
\begin{equation}\label{eq:optimbu}
\begin{split}
k_{\rm safe}(x) = &\argmin_{u\in \mcl{U}} \|u-k_{\rm p}(x)\|^2\\
&\text{s.t.}\,\eqref{eq:hb1},\,\eqref{eq:hb2}.
\end{split}
\end{equation}
Since \(\tau\) varies continuously over \([0,T]\),~\eqref{eq} imposes infinitely many constraints. In implementation, we enforce them on a finite grid assumed dense enough to ensure satisfaction over the full horizon.

Although the backup CBF framework has been applied successfully to various control problems, the size of the expanded set~$\Cset_{\rm I}$ depends on the choice of backup set -- backup controller pair $(\Cset_{\rm B}, k_{\rm b})$ \cite{Chen2021bCBF,vanWijk2026Gen}. Section~\ref{sec:method} derives an approach to synthesize a new pair from a backward reachable set. The next subsection introduces backward reachability.

\subsection{Backward Reachability}\label{ss:BRS}
A safe region described by state constraints can also be represented using a so-called target tube $\Omega_{t,0}^{r}$. The target tube is defined by a user-specified function $r(t,x)$ and can be used to define an associated backward reachable set (BRS) 
\begin{equation*}
   \Omega_\text{BRS}:=\{x(t_0) \in \mcl{X}: \exists u, \text{s.t.} ~ x(t) \in \Omega_{t,0}^r, \forall t\in [t_0,T] \}. 
\end{equation*} 
The BRS consists of all initial conditions from which an admissible controller can keep the system within the target tube. Thus, it characterizes the set of states from which safety can be maintained by an admissible controller.
The associated state feedback controller~$k:\mcl{X} \rightarrow \mcl{U}$ with~$u(t) = k(x(t))$ is time-invariant and memoryless. We further assume that there are polytopic constraints on the control input, i.e.,~$u$ is confined to the set~%$\mcl{U}(t,x):=\{u\in \R^m: A(t,x)u\le b(t,x)\}$. 
\begin{equation}\label{eq:U}
    \mcl{U}:=\{u\in \R^m: Au\le b\},
\end{equation}
where $A\in \R^{n_{\rm p}\times m}$ and $b\in\R^{n_{\rm p}}$ and $n_{\rm p}$ is the number of constraints on~$u$, and and \(\le\) is interpreted componentwise.

An inner approximation of~$\Omega_\text{BRS}$  can be obtained from a dissipation argument using sublevel sets of reachability storage functions ${V(t,x), ~ V : \R \!\times \!\mcl{X} \to \R}$. The next result gives local synthesis conditions for $k$ under control constraints.

\begin{proposition}[\!\!\cite{yin2021backward}]\label{thm:brs}
Consider the system~\eqref{eq:sys_dynamics}, initial time~$t_0$, terminal time~$T\ge t_0$, polytopic input control constraints defined by~$A\in \R^{n_{\rm p}\times m}$ and~$b\in \R^{n_{\rm p}}$, function~$r$ and associated target tube~$\Omega_{t,0}^r$ and~$\gamma\in\R$. If there exists a~$C^1$ function~$V:\R \times \mcl{X} \rightarrow \R$, and a control law~$k: \mcl{X} \rightarrow \mcl{U}$ that is locally Lipschitz in~$x$, such that for all~$t\in [t_0, T]$, the following constraints hold:
\begin{subequations}
    \begin{align}
        \begin{split}
        &\Omega^V_{t, \gamma} \subseteq \left\{x \in \mcl{X} : \frac{\partial V(t, x)}{\partial t} + \frac{\partial V(t, x)}{\partial x} \bigl(f(x) \right., \\
        &\qquad \;\; \left. \phantom{\frac{\partial}{\partial}} + g(x)k(x)\bigr) \leq 0 \right\},
        \end{split} \\
        &\Omega_{t,\gamma}^V \subseteq \Omega_{t,0}^r,\\
        &\Omega_{t,\gamma}^V \subseteq \{x \in \mcl{X} : Ak(x)\le b\},
    \end{align}
\end{subequations}
then, under ${u=k(x)}$, any trajectory with initial condition~$x(t_0)\in\Omega_{t_0,\gamma}^V$, satisfies~$x(t)\in\Omega_{t,0}^r$ for all~$t\in[t_0,T]$.
\end{proposition}
Thus, backward reachability under control constraints yields a reachable set--controller pair. We will exploit this fact in Section~\ref{sec:method} to calculate a backup set--backup controller pair.

Recall that the system~\eqref{eq:sys_dynamics} and the control law are assumed to be polynomials. Thus, we make the same assumption on the reachability storage function, i.e., $V\in \R[t,x]$. We will exploit this assumption in the next section to derive a polynomial optimization problem to synthesize a backup set--backup controller pair.

\section{Main Result}
This section constructs an enlarged backup set using SOS-based backward reachability. We first synthesize an SOS-certified initial set and associated controller, then provide conditions under which this set defines a valid backup pair for the bCBF framework.

\label{sec:method}
\subsection{Offline Sum-of-Squares Controller Synthesis}
Let $T_1 > 0$ be a fixed time horizon and let~$\Cset_{\rm B}$ denote the invariant set defined as in~\eqref{eq:Cb}.
We seek a reachability storage function
$V : [0,T_1] \times \mcl{X} \to \R$,
a control law $k : \mcl{X} \to \mcl{U}$,
and a $\gamma \in \R$
such that for all $(t,x) \in [0,T_1]\times \mcl{X}$:
\begin{align}
\label{eq:offline_dissipation}
\begin{split}
&V(t,x) \le \gamma \\
&\;\Rightarrow\; 
\frac{\partial V(t,x)}{\partial t}
+ \frac{\partial V(t,x)}{\partial x}
\big(f(x) + g(x)k(x)\big) \le 0,
\end{split}
\\
&V(T_1,x) \le \gamma
\;\Rightarrow\;
x \in \Cset_{\rm B},
\label{eq:offline_terminal}
\\
&V(0,x) \le \gamma
\;\Rightarrow\;
x \in \Cset_{\rm S},
\label{eq:offline_initial_safe}
\\
&V(t,x) \le \gamma
\;\Rightarrow\;
V(0,x) \le \gamma
\label{eq:offline_TV}
\\
&x \in \Cset_{\rm B}
\;\Rightarrow\;
V(0,x) \le \gamma,
\label{eq:ensure_BRS_enlargement}
\\
&V(t,x) \le \gamma
\;\Rightarrow\;
k(x) \in \mcl{U}.
\label{eq:offline_input}
\end{align}
We note that condition \eqref{eq:offline_terminal} is the terminal target set condition in the standard finite horizon backward reachable set construction, i.e., ${\Omega_{T_1,\gamma}^V \subseteq \mathcal{C}_{\rm B}}$. By using the target tube, safety over the horizon could be imposed directly by requiring 
%\begin{equation*}
${V(t,x)\le \gamma \Rightarrow x \in \mathcal C_{\rm S}}.$
%\end{equation*}
Here, we instead enforce this requirement through the two conditions \eqref{eq:offline_initial_safe} and \eqref{eq:offline_TV}. This decomposition is stronger than the direct target tube safety condition, but it is imposed so that the initial sublevel set $\Omega_{0,\gamma}^V$ is itself safe by \eqref{eq:offline_initial_safe}. Moreover, together with \eqref{eq:offline_TV} and \eqref{eq:ensure_BRS_enlargement}, this allows $\Omega_{0,\gamma}^V$ to serve as a backup set in
Theorem~\ref{thm:validity}.

Associated with the control law $k(x)$, $\Omega_{0,\gamma}^V$ is an inner approximation of the BRS of $\Cset_{\rm B}$ over $[0,T_1]$, provided that \eqref{eq:offline_terminal}--\eqref{eq:offline_input} hold. Recall that the control inputs are confined to a polytopic set $\mcl{U}$ as defined in \eqref{eq:U}. Combining the constraints \eqref{eq:offline_dissipation}--\eqref{eq:offline_input} yields a synthesis optimization problem. As a larger inner-approximation is preferable, the volume of $\Omega_{0,\gamma}^V$ becomes the objective to be maximized. \\
\textit{High-level optimization problem: ($hi\text{-}opt$)}
\begin{align*}
    &\sup_{V, k} \; \vol(\Omega_{0,\gamma}^V) \\
    & \text{s.t. } \Omega^V_{t, \gamma} \subseteq \left\{ x \in \mcl{X} : \frac{\partial V(t, x)}{\partial t} + \frac{\partial V(t, x)}{\partial x} \bigl(f(x) \right.\\
    &\qquad \qquad \left. \phantom{\frac{\partial}{\partial}} + g(x)k(x)\bigr) \leq 0 \right\},\, \forall t \in [0, T_1], \tag{A.1} \\
    % \displaybreak
    & \quad \; \Omega^V_{T_1,\gamma} \subseteq \Omega^{-h_{\rm b}}_0,
    \tag{A.2}\\
    & \quad \; \Omega^V_{0,\gamma} \subseteq \Omega^{-h}_0, \tag{A.3}\\
    & \quad \; \Omega_{t,\gamma}^V \subseteq \Omega_{0, \gamma}^V, \forall t \in [0, T_1], \tag{A.4}\\
    & \quad \; \Omega^{-h_{\rm b}}_0 \subseteq \Omega^V_{0, \gamma}, \tag{A.5}\\
    & \quad \; \Omega^V_{t,\gamma} \subseteq \{ x \in \mcl{X} : Ak(x) \leq b \},\, \forall t \in [0, T_1]. \tag{A.6}
\end{align*}

\medskip
\noindent
The system dynamics, reachability storage function, and control law are polynomials, i.e., $f \in \R^n[x]$, $g \in \R^{n \times m}[x]$, $V \in \R[t, x]$, and $k \in \R^m[x]$.
By applying the S-procedure and the fact that ${l(t) := t(T_1 - t)}$ is non-negative on $t \in [0, T_1]$, we can convert $hi\text{-}opt$ into a computationally tractable optimization problem based on the sufficient conditions given below. \\
\textit{Optimization problem: ($sosopt$)} %Fix $\epsilon > 0$
\begin{align*}
    & \sup_{V, k, s} \; \vol(\Omega^V_{0, \gamma}) \\
    & \text{s.t. } s_2(t, x), s_3(t, x), s_6(t, x), s_7(t, x) \in \Sigma[t, x], \\
    & \quad \; (s_4(x) - \epsilon), (s_5(x) - \epsilon), s_8(x) \in \Sigma[x], \\
    & \quad \; s_{9,i}(t, x), s_{10,i}(t, x) \in \Sigma[t, x], \forall i = 1, \dots, n_p, \\
    & \quad \; k(x) \in \R^m[x], V(t, x) \in \R[t, x], \tag{B.1} \\
    & \quad \; -\left( \frac{\partial V(t, x)}{\partial t} + \frac{\partial V(t, x)}{\partial x} (f(x) + g(x)k(x)) \right) \\
    & \quad \;  - s_2(t, x)l(t) + s_3(t, x)(V(t, x) - \gamma) \in \Sigma[t, x], \tag{B.2} \\
    & \quad \; s_4(x)h_{\rm b}(x) + V(T_1, x) - \gamma \in \Sigma[x], \tag{B.3} \\
    & \quad \; s_5(x)h(x) + V(0, x) - \gamma \in \Sigma[x], \tag{B.4} \\
    & \quad \; -(V(0, x) - \gamma) - s_6(t, x)l(t) \\
    & \quad \; + s_7(t, x)(V(t, x) - \gamma) \in \Sigma[t, x],
    \tag{B.5} \\
    & \quad \; -(V(0, x) - \gamma) - s_8(x)h_{\rm b}(x) \in \Sigma[x], \tag{B.6} \\
    & \quad \; b_i - A_ik(x) + s_{9,i}(t, x)(V(t, x) - \gamma) \\
    & \qquad - s_{10,i}(t, x)l(t) \in \Sigma[t, x],\, \forall i = 1, \dots, n_p, \tag{B.7}
\end{align*}
where the positive constant $\epsilon$ ensures that $s_4(x)$ and $s_5(x)$ are uniformly bounded away from 0. Here, $A_i$ denotes the $i$-th row of $A$ and $b_i$ denotes the $i$-th element of $b$. The constraints (B.2), (B.5), and (B.7) contain the bilinear terms $(\frac{\partial V}{\partial x}, k)$, $(s_3,V)$, $(s_7,V)$, and $(s_{9,i},V)$, making the problem nonconvex. We therefore solve \textit{($sosopt$)} iteratively, following \cite{yin2021backward}. Algorithm~\ref{alg:sositer} summarizes the implementation.
\begin{algorithm}
\caption{Iterative method}
\label{alg:sositer}
\begin{algorithmic}
\STATE \textbf{Input:} function $V^0$ such that constraints (B.2)--(B.7) are feasible by proper choice of $s, k, \gamma$.
%\STATE \textbf{Output:} $(k, \gamma, V)$ such that with the volume of $\Omega_{0,\gamma}^V$ having been enlarged. 
\STATE \textbf{Output:} $(k, \gamma, V)$ for enlarged volume of $\Omega_{0,\gamma}^V$.
\FOR{$j = 1 : N_{\rm iter}$} 
    \STATE \textbf{$\gamma$-step:} 
    \STATE Decision variables $(s, k, \gamma)$. 
    \STATE Maximize $\gamma$ subject to (B.1)--(B.7) using $V = V^{j-1}$. \\ This yields $(s_3^j, s_7^j, s_{9,i}^j,  k^j)$ and optimal reward $\gamma^j$. 
    \medskip
    \STATE \textbf{$V$-step:} 
    \STATE Decision variables $(s_1, s_2, s_4, s_5, s_6, s_8, s_{10,i}, V)$.
    \STATE Maximize the feasibility (analytic center described as in Remark 5 of \cite{yin2021backward}) subject to (B.1)--(B.7) as well as $s_1(x) \in \Sigma[x]$, and 
    \begin{align*}
    &-(V(0, x) - \gamma^j) \\
    &+ s_1(x)(V^{j-1}(0, x) - \gamma^j) \in \Sigma[x], \tag{B.8}
    \end{align*}
    
    \STATE using $(s_3 = s_3^j, s_7 = s_7^j, s_{9,i} = s_{9,i}^j, k = k^j, \gamma = \gamma^j)$. This yields $V^j$. 
\ENDFOR
\end{algorithmic}
\end{algorithm}
\begin{remark}
In the~$\gamma$-step of Algorithm~\ref{alg:sositer}, only the terms $(s_3, \gamma)$, $(s_7, \gamma)$, and $(s_{9,i}, \gamma)$ are bilinear, and $\gamma$ is the objective value. The $\gamma$-step is a generalized SOS problem, i.e., a series of SOS feasibility problems. For fixed multipliers, the~$\gamma$-step reduces to a quasi-convex problem and the step-wise global optimum can be computed via bisection over $\gamma$. Note that while the individual steps are convex, the iterative algorithm, in general, only converges to a local optimum of the SOS relaxation.
\end{remark}
\begin{remark}
The examples in Section~\ref{sec:numex} use the time-independent function~$-h_{\rm b}$ as the initial guess for the reachability storage function $V^0$. This initial guess inherently satisfies the constraints (A.2 - A.5) with $\gamma = 0$.
\end{remark}

\subsection{Extended Backup Controller}
We want to use $\Omega_{0,\gamma}^V$ obtained from Algorithm~\ref{alg:sositer} as an enlarged backup set for $\Cset_{\rm B}$. Thus, as done in \cite{vanWijk2026Gen}, we define the corresponding extended backup controller as follows:
\begin{equation}
k_{\rm e}(x)
:=
\begin{cases}
k_{\rm b}(x), & h_{\rm b}(x) \geq 0, \\
{\rm sat}_{\mcl{U}}(k(x)), & h_{\rm b}(x) < 0,
\end{cases}
\label{eq:extended_backup_policy}
\end{equation}
where ${{\rm sat}_{\mcl{U}} : \R^m \rightarrow \mathcal{U}}$ is a smooth saturation function.
This controller can be made smooth by interpolating at the boundary of $h_{\rm b}$ as done in  \cite{vanWijk2026Gen}.

We also define $\phi_{\rm e}(\tau,x)$ to be the flow of the extended backup system $\dot{x} \!= f_{\rm e}(x) \!:=\! f(x) + g(x)k_{\rm e}(x)$, which satisfies
\begin{equation}
\tfrac{\partial }{\partial \tau}\phi_{\rm e}(\tau,x)
=
f_{\rm e}\big(\phi_{\rm e}(\tau,x)\big),
\quad
\phi_{\rm e}(0,x)=x.
\label{eq:extended_backup_flow}
\end{equation}

The next result shows that the backup set method can be applied with the BRS-based set $\Omega_{0,\gamma}^V$ and the extended backup controller $k_{\rm e}$ in \eqref{eq:extended_backup_policy}.
\medskip
\begin{theorem}
\label{thm:validity}
Let $V$, $k$, and $\gamma$ be a feasible solution of Algorithm~\ref{alg:sositer}.
Then the set $\Omega_{0,\gamma}^V$ is controlled invariant,
and the extended backup controller defined in \eqref{eq:extended_backup_policy}
renders $\Omega_{0,\gamma}^V$ forward invariant along
$\dot{x} \!=\! f_{\rm e}(x)$, i.e.,
\begin{align}
    x \in \Omega_{0,\gamma}^V
  \;\Rightarrow\;
  \phi_{\rm e}(\tau,x) \in \Omega_{0,\gamma}^V,\, \forall \tau \geq 0.
\end{align}
\end{theorem}

\begin{proof}
We consider two cases: %based on the initial condition.

%\medskip
\noindent\textit{Case 1.} ($x \in \Cset_{\rm B}$):
Since $h_{\rm b}(x)\geq 0$, we have $k_{\rm e}(x)=k_{\rm b}(x)$.
By assumption, $k_{\rm b}$ renders $\Cset_{\rm B}$ forward invariant.
Moreover, constraint (B.6) implies (A.5), yielding
$\Cset_{\rm B} \subseteq \Omega_{0,\gamma}^V$.
Hence,
\begin{equation*}
\phi_{\rm e}(\tau,x) \in \Cset_{\rm B} \subseteq \Omega_{0,\gamma}^V,\, \forall \tau \geq 0.
\end{equation*}

\medskip
\noindent\textit{Case 2.} ($x \in \Omega_{0,\gamma}^V \setminus \Cset_{\rm B}$):
Since (B.3) implies 
$\Omega_{T_1,\gamma}^{V} \subseteq \Cset_{\rm B}$,
 the flow enters $\Cset_{\rm B}$ within time $T_1$.
Hence, the first entry time into $\Cset_{\rm B}$ is well-defined as
\begin{equation*}
    \tau^{*} := \min\bigl\{\tau \in (0,T_1] : h_{\rm b}\bigl(\phi_{\rm e}(\tau,x)\bigr)\geq 0\bigr\}.
\end{equation*}

For $\tau \in [0, \tau^{*})$, we have $h_{\rm b}\bigl(\phi_{\rm e}(\tau,x)\bigr) < 0$.
Constraints (B.2) and (B.7) imply (A.1) and (A.6), respectively.
Combining these results, the flow satisfies
\begin{equation*}
\phi_{\rm e}(\tau,x) \in \Omega_{\tau,\gamma}^V,
\end{equation*}
with $k_{\rm e}(\phi_{\rm e}(\tau,x)) = \operatorname{sat}_{\mcl{U}}(k(\phi_{\rm e}(\tau,x))) = k(\phi_{\rm e}(\tau,x)),\, \forall \tau \in [0, \tau^{*})$.
Furthermore, (B.5) implies (A.4), yielding
$\Omega_{\tau,\gamma}^{V}\subseteq\Omega_{0,\gamma}^V$ for all $\tau \in [0,T_1]$.
Hence,
\begin{equation*}
\phi_{\rm e}(\tau,x)\in\Omega_{0,\gamma}^V,\, \forall \tau \in [0,\tau^{*}).
\end{equation*}

For $\tau \geq \tau^{*}$, since $\phi_{\rm e}(\tau^*,x) \in \Cset_{\rm B}$, applying \textit{Case 1} from time $\tau^{*}$ onward yields
\begin{equation*}
\phi_{\rm e}(\tau,x)
= \phi_{\rm e}\bigl(\tau-\tau^{*},\,\phi_{\rm e}(\tau^{*},x)\bigr)
\in \Cset_{\rm B} \subseteq \Omega_{0,\gamma}^V,\, \forall \tau \geq \tau^{*}.
\end{equation*}
\end{proof}

When the backup set method (bCBF) described in
Section~\ref{sec:backup} is applied with the pair
$(\Omega_{0,\gamma}^V,k_{\rm e})$, the corresponding controlled
invariant set $\Cset_{\rm I}^\text{e} \subseteq \Cset_{\rm S}$ is defined as
\begin{equation}
\Cset_{\rm I}^\text{e} :=
\left\{
x \in \mcl{X} \;\middle|\;
\begin{aligned}
& h\big(\phi_{\rm e}(\tau,x)\big) \geq 0,
\, \forall \tau \in [0,T], \\
& \gamma - V\big(0, \phi_{\rm e}(T,x)\big) \geq 0
\end{aligned}
\right\},
\label{eq:extended_CI}
\end{equation}

If $k_{\rm e}(x)$ is $C^1$ (e.g., by smoothly interpolating at the boundary of $h_{\rm b}$ \cite{vanWijk2026Gen}), then the sensitivity matrix $\frac{\partial \phi_{\rm e}(\tau,x)}{\partial x}$ exists, which enables the application of the same procedure described in Section~\ref{sec:backup}. As the set~$\Cset_{\rm I}^\text{e}$ has the same properties as the expanded backup set~$\Cset_{\rm I}$, we can calculate the minimally invasive safe control signal~$k_{\rm safe}$ in the same way as in~\eqref{eq:optimbu}.

\section{Numerical Experiments}
\label{sec:numex}
\subsection{Double Integrator}

\begin{figure}[t]
\centering
\includegraphics[width = 0.48\textwidth]{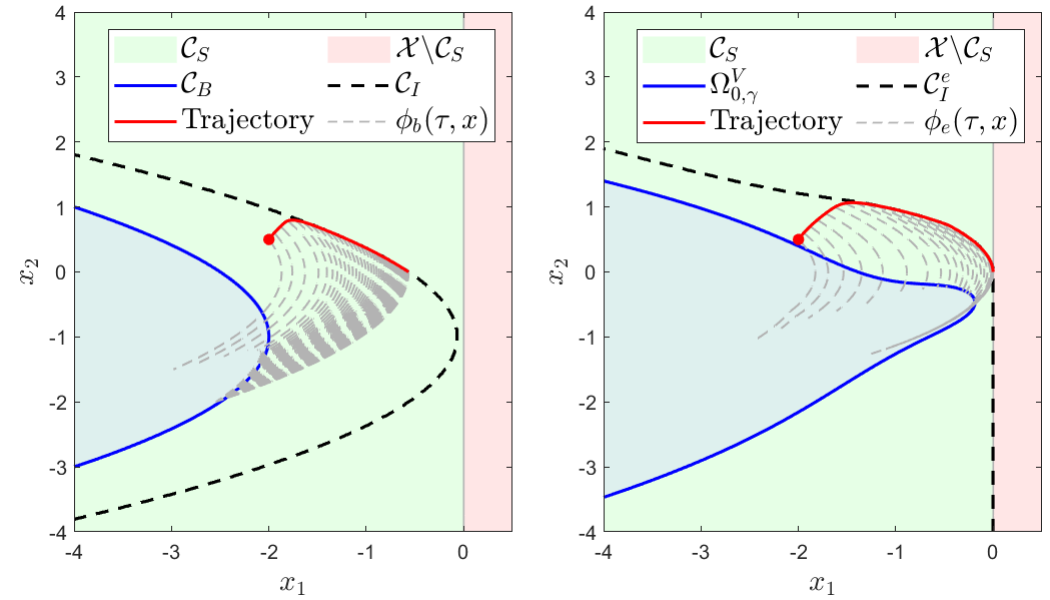}
\includegraphics[width = 0.48\textwidth]{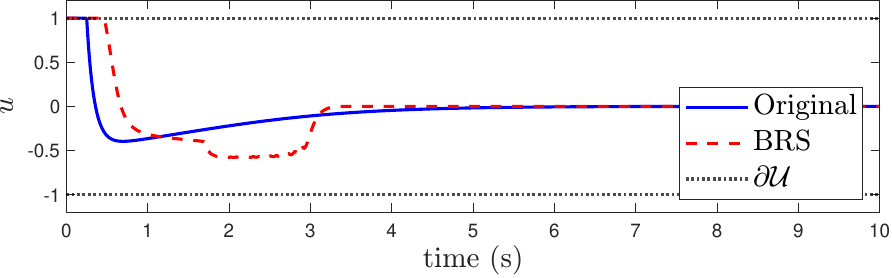}
\caption{Simulation results from the bCBF-QP for the double integrator. The results obtained using the original backup set $\Cset_{\rm B}$ and the backup controller $k_{\rm b}(x)$ are shown in the top-left plot, while those obtained using the inner-approximated BRS $\Omega_{0,\gamma}^V$ and the extended backup controller $k_{\rm e}(x)$ are shown in the top-right plot. The bottom plot shows the safe control inputs computed by the original and BRS-based bCBF-QPs.
}
\label{fig:DBIntSimcomp}
\vspace{-3mm}
\end{figure}

Consider the double integrator
\begin{align}
\label{eq}
\dot{x} = \bmtx x_2 & u\emtx^{\top},
\end{align}
where $x=[x_1,x_2]^\top$ contains the position $x_1$ and velocity $x_2$, and $u\in\mcl{U}=[-1,1]$. The safe set $\Cset_{\rm S}:=\{x\in\R^2:-x_1\geq0\}$ is the left half-plane. We use the primary controller $k_{\rm p}(x)=1$, which drives \eqref{eq} toward the unsafe right half-plane. The backup controller $k_{\rm b}(x)=-1$ drives the system into $\Cset_{\rm B}:=\{x\in\R^2:-x_1-2-0.5(x_2+1)^2\geq0\}$ and renders $\Cset_{\rm B}$ forward invariant. The bCBF-QP horizon $T$ and Algorithm~\ref{alg:sositer} horizon $T_1$ are both set to $2\,({\rm s})$. The polynomial degrees of $V$, $k$, and $s$ in Algorithm~\ref{alg:sositer} are six, two, and two, respectively. We use $\epsilon=1 \times 10^{-4}$.

Fig.~\ref{fig:DBIntSimcomp} compares the simulation results for the original and BRS-based bCBF-QPs. The inner-approximated BRS $\Omega_{0,\gamma}^V$ contains
$\Cset_{\rm B}$ by construction, and the
resulting controlled invariant set $\Cset_{\rm I}^\text{e}$ is noticeably larger than $\Cset_{\rm I}$. The backup flows $\phi_{\rm b}(\tau,x)$ and $\phi_{\rm e}(\tau,x)$ show that trajectories starting within $\Cset_{\rm I}$ and $\Cset_{\rm I}^\text{e}$ are steered into $\Cset_{\rm B}$ and $\Omega_{0,\gamma}^V$, respectively, within the finite horizon $T$. The bottom plot compares the safe control inputs obtained for both approaches. The original bCBF-QP starts modifying the control signal shortly after the start of the simulation, while the BRS-based bCBF-QP allows the primary controller to remain active longer before the safety constraint becomes active. Both cases fulfill the control input constraint $\mcl{U} = [-1,1]$ at all times.

\begin{figure}[t]
\centering
\includegraphics[width = 0.48\textwidth]{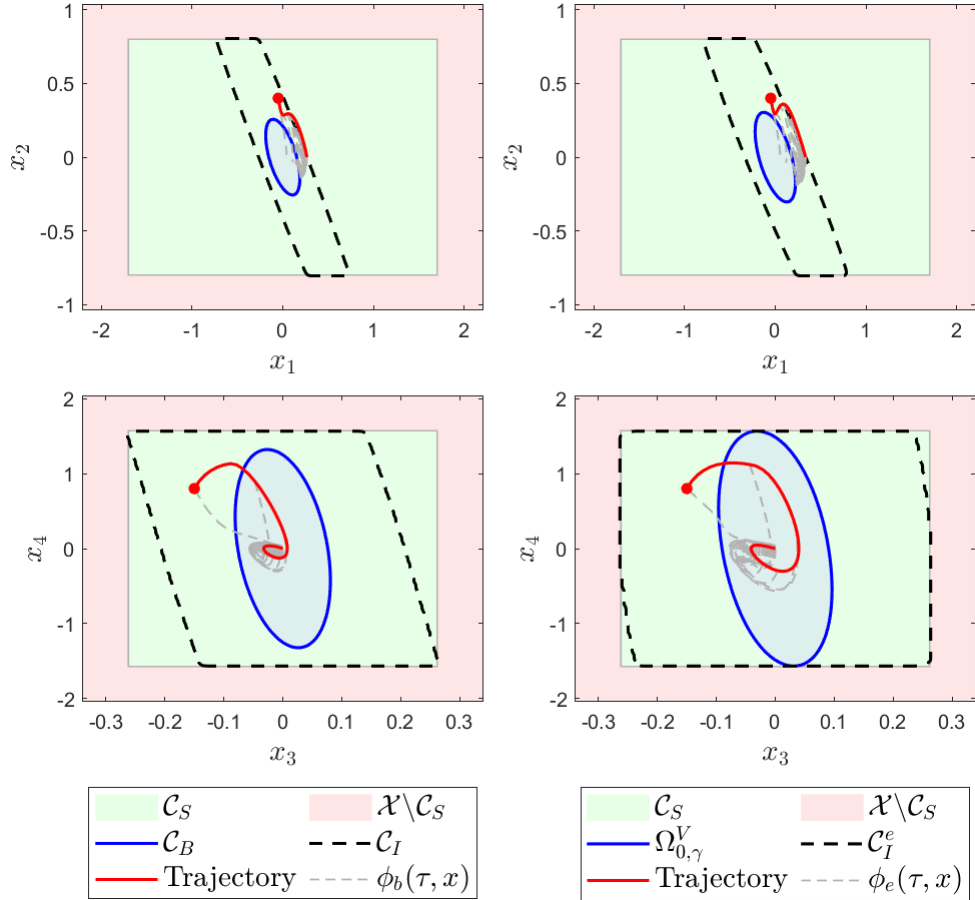} \\[0.1cm]
\includegraphics[width = 0.48\textwidth]{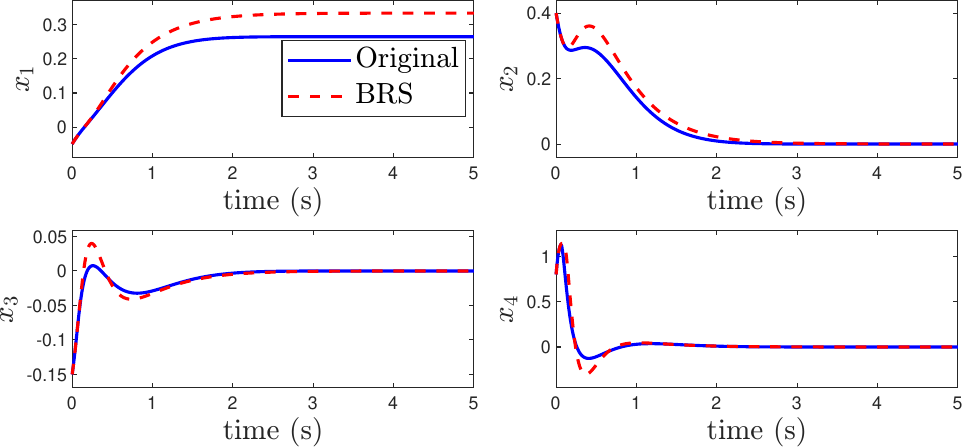} \\[0.1cm]
\includegraphics[width = 0.48\textwidth]{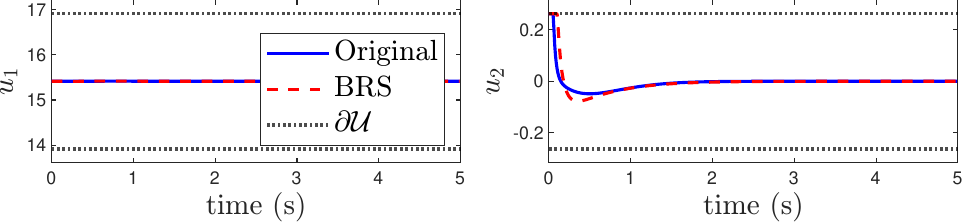}
\caption{Simulation results from the bCBF-QP for the quadcopter. %, where the vertical dynamics are neglected, resulting in a reduced 4-state model. 
The first two rows compare the results obtained using the original backup set $\Cset_{\rm B}$ and the backup controller $k_{\rm b}(x)$ on the left with those obtained using the inner approximation of the BRS $\Omega_{0,\gamma}^V$ and the extended backup controller $k_{\rm e}(x)$ on the right. These rows show slice plots in the $(x_1,x_2)$ and $(x_3,x_4)$ subspaces, with $x_3=x_4=0$ and $x_1=x_2=0$, respectively. The middle two rows show the state trajectories, and the bottom row shows the safe control inputs computed by the original and BRS-based bCBF-QPs.
}
\label{fig:RdQuadSimcomp}
\vspace{-3mm}
\end{figure}

\subsection{Quadrotor with Vertical Dynamics Neglected}
Consider the horizontal dynamics of a quadcopter~\cite{yin2021backward}:
\begin{align}
\label{eq:rdquad}
\underbrace{
\bmtx 
\dot{x}_1 \\ \dot{x}_2 \\ \dot{x}_3 \\ \dot{x}_4 
\emtx}_{\dot{x}} \!=\! 
\underbrace{
\bmtx 
x_2 \\ 0 \\ x_4 \\ -d_0 x_3 - d_1 x_4 
\emtx}_{f(x)} \!+\! 
\underbrace{
\bmtx 
0 & 0 \\ 
K\sin(x_3) & 0 \\ 
0 & 0 \\ 
0 & n_0 
\emtx}_{g(x)}
\underbrace{
\bmtx 
u_1 \\ u_2 
\emtx}_{u},
\end{align} 
where the state vector $x = [x_1,x_2,x_3,x_4]^{\top}$ consists of horizontal position $x_1$ in $\text{(m)}$, horizontal velocity $x_2$ in $\text{(m/s)}$, roll angle $x_3$ $\text{(rad)}$, and roll rate $x_4$ in $\text{(rad/s)}$. The bounded inputs are the total thrust $u_1 \in [-1.5, 1.5] + g/K$ and desired roll angle $u_2 \in [-\pi/12, \pi/12]$. The gravitational acceleration is $g = 9.8\,\text{m/s}^2$. The system parameters are given by $K = 0.89/1.4$, $d_0 = 70$, $d_1 = 17$, and $n_0 = 55$. We approximate the term $\sin(x_3)$ by $x_3 - 0.166x_3^3$.
The safe set is defined as $\Cset_{\rm S} := \{x \in \R^4 : |x_1| \leq 1.7, |x_2| \leq 0.8, |x_3| \leq \pi/12, |x_4| \leq \pi/2\}$. After linearizing \eqref{eq:rdquad}, the backup set $\Cset_{\rm B}$ and the corresponding backup controller $k_{\rm b}(x)$ are defined by a level set of a quadratic Lyapunov function centered at the equilibrium point $(x_{\rm e}, u_{\rm e}) = (\mbf{0}_4, [g/K, 0]^{\top})$ and a LQR state-feedback controller \cite{Gurriet2020}. The primary controller $k_{\rm p}(x)$ is chosen as $[ g/K , \pi/12]^{\top}$. The finite horizons of the backup CBF-QP ($T$) and in Algorithm~\ref{alg:sositer} ($T_1$) are both set to $1\,(\rm{s})$. The degrees of $V$, $k$, and $s$ used in Algorithm~\ref{alg:sositer} are four, two, and two, respectively. We use $\epsilon = 1 \times 10^{-4}$.

Fig.~\ref{fig:RdQuadSimcomp} compares the original and BRS-based bCBF-QPs. The first two rows show two-dimensional slices of the certified sets in the $(x_1,x_2)$ and $(x_3,x_4)$ subspaces (off-axis states fixed to zero)
together with the closed-loop trajectories and backup flows. By construction, the inner-approximated BRS $\Omega_{0,\gamma}^V$ contains $\mathcal{C}_{\rm B}$ resulting in an extended controlled invariant set $\Cset_{\rm I}^\text{e}$ that is noticeably larger than $\Cset_{\rm I}$. The backup flows $\phi_{\rm b}(\tau,x)$ and $\phi_{\rm e}(\tau,x)$ confirm that trajectories starting within $\Cset_{\rm I}$ and $\Cset_{\rm I}^\text{e}$ are driven into $\Cset_{\rm B}$ and $\Omega^V_{0,\gamma}$, respectively, within the finite horizon $T$. The controlled invariant sets are visualized as two-dimensional slices obscuring their four-dimensional geometry, which makes some trajectories appear to exit $\Cset_{\rm I}$ and $\Cset_{\rm I}^\text{e}$. However, the plots of the individual state trajectories confirm this as a visual artifact. Finally, 
% the bottom row shows that 
both controllers satisfy the control input constraints $u_1 \in [-1.5, 1.5] + g/K$ and $u_2 \in [-\pi/12, \pi/12]$ throughout the simulation, with the BRS-based bCBF-QP allowing the primary controller to be active slightly longer.

\section{Conclusion and Future Work}

We presented an SOS-based method for synthesizing enlarged backup sets for input-constrained backup control barrier functions. By certifying a set as a valid backup set under a piecewise backup controller, the approach reduces conservatism in bCBF safety filtering. Future work will address scalability through scaled diagonally dominant sum-of-squares (SDSOS) optimization \cite{ahmadi2019dsos} to compute control invariant sets for higher-dimensional polynomial systems.

\bibliographystyle{IEEEtran}
\bibliography{references}

\end{document}